\documentclass{article}\textheight9in\textwidth6.5in
\begin{document}

 \newcommand\II{\makebox[0pt][l]{\hspace*{1pt}I}}
 \newcommand\N{{\mbox{\bf\II N}}} \newcommand\Q{{\mbox{\bf\II Q}}}
 \newcommand\qed{\rule{1em}{1.5ex}} \newcommand\ov\overline
 \newcommand\SKH{\cal H} \newcommand\SKS{\cal S} \newcommand\SKF{\cal F}
 \newcommand\SKJ{\cal J} \newcommand\SKN{\cal N} \newcommand\SKM{\cal M}
 \newcommand\I{{\mbox{\bf I}}}
 \renewcommand\S{{\mbox{\bf S}}} \newcommand\M{{\mbox{\bf M}}}
 \newcommand\D{{\mbox{\bf D}}} \newcommand\K{{\mbox{\bf K}}}
 \newcommand\m{{\mbox{\bf m}}} \renewcommand\d{{\mbox{\bf d}}}

 \newtheorem{thr} {Theorem} \newtheorem{lem} {Lemma} \newtheorem{nt} {Note}
 \newtheorem{prp}{Proposition} \newtheorem{dfn} {Definition}
 \newtheorem{Postulate}{Postulate} \newtheorem{law}{Law}
 \newtheorem{cor} {Corollary}

 \title {\null\vspace{-1in} Randomness Conservation Inequalities;\\
	Information and Independence in\\ Mathematical Theories\thanks
{Information and Control, 61(1):15-37, April 1984}\thanks {Supported in 1978-83
by NSF grants \#: MCS 77-19754, 81-04211 and 83-04498.}} \author {Leonid
A.~Levin\\
 Boston University, Boston Massachusetts and\\ Massachusetts Institute of
Technology, Cambridge, Massachusetts.\thanks {Correspondence should be
addressed to the author at 150-3 Kenrick Street, Boston, MA 02135.}}

 \date{} \maketitle \begin{abstract}
 The article develops further Kolmogorov's Algorithmic Complexity Theory. The
definition of Randomness is modified to satisfy strong invariance properties
(conservation inequalities). This allows definitions of concepts such as Mutual
Information in individual infinite sequences. Applications to several areas,
like Probability Theory, Theory of Algorithms, Intuitionistic Logic are
considered. These theories are simplified substantially with the postulate that
the objects they consider are independent of (have small mutual information
with) any sequence specified by a mathematical property. \end{abstract}

\vspace{1pc}\centerline{\bf\large I. ALGORITHMIC INFORMATION}

 \addtocounter{section}{-1} \section{\label{0} Initial Remarks}
	\subsection{\label{0.1} Introduction}

 Recursive function theory provides analogues of many concepts of classical
analysis by requiring countable sets considered to be recursively enumerable
(r.e.). The analogy is quite close: intrinsically non-algorithmic methods are
rare in mathematics. Moreover, the general Theory of Algorithms is very similar
to descriptive set theory. There is, however, an important exception in the
existence of universal algorithms. The set of all (countable) sets of integers
is uncountable while the set of r.e. sets is r.e.. This rather abstract
difference opens, however, new analytical possibilities having no analogies in
``non-algorithmic" analysis. Let us illustrate this with a simple but important
example.

 Let $l_1\subset\Re^\N$ be the space of all absolutely summable real sequences:
$p \in l_1$, iff $\sum |p(x)|<\infty$. Its recursive analogue $\ov{l_1}\subset
l_1$ consists of elements of $l_1$ whose subgraph $\{(r,x):\ p(x)> r\in \Q\}$
is r.e. It is known in calculus that no element is maximal in $l_1$ within a
constant factor: $\forall p\in l_1\exists q\in l_1\lim q(x)/p(x)=\infty$. In
contrast to this $\ov{l_1}$ has an ``absorbing" element $\m$ (a {\em Universal
Measure}) such that $$\forall q\in\ov{l_1}\sup(q(x)/\m(x))<\infty.
	\mbox{\hspace{4pc} Here }\m(x)=\sum r_i(x)/2i^2,$$
 where $\{r_i\}$ is an r.e. family of all non-negative r.e. sequences with the
sum bounded by $1$.

 {\em Complexity} $\K(x)=-\lfloor\ln\m(x)\rfloor$ is closely related to the
length of a shortest program generating $x$ in an optimal language discovered
in [Kolmogorov 65], [Solomonoff]. Their work originated an invariant approach
to Information Theory, foundations of Probability Theory, Inductive Inference
and a number of other areas. Any function of integers, invariant with respect
to all recursive transformations, is a constant. However {\K} is {\em
approximately} invariant, i.e. $\sup(\K(\varphi(x))-\K(x))<\infty$ for any
recursive $\varphi$. Concepts like {\em Mutual Information} $\I(x:y)=\K(x)+
\K(y)-\K(x,y)$ or {\em Deficiency of Randomness} (with respect to a
\mbox{measure $\mu$):} $|\ln\mu(x)|-\K(x)$ also have attractive invariance
properties suitable for many applications. These idea of Algorithmic
Information Theory, are based on analytical features arising in the recursive
analogues of some spaces of classical analysis, as in the above example.

 We hope to introduce the reader to the general spirit of the theory by
choosing a particular problem and developing the concepts needed for its
solution. The problem used for organizing this work is to formalize, justify,
and apply the following physical principle:

 Let $R$ be the reference to a physical process generating sequence $\alpha_R$.
Let $P$ be a (non-recursive) mathematical property specifying sequence
$\beta_P$. The lengths of $R$ and $P$ may be {\em negligible} compared to the
informational content of $\alpha_R$ and $\beta_P$. E.g., $R$ may be the
bibliographical reference to a book $\alpha_R$ and $P(\beta)$ may be ``$\beta$
is the first sequence which cannot be generated by a program of $<10^{10}$
bits". \mbox{Then predicted is:}

\begin {Postulate} [Independence] The sequences $\alpha_R$ and $\beta_P$ are
independent, i.e. $I(\alpha_R:\beta_P) < |R|+|P|$.\end{Postulate}

 A special case of $\alpha=\beta$ gives a ``finitary Church's Thesis": Every
``physically existing" sequence must be ``finitary recursive", i.e. have
approximately as short recursive expression as any of its non-recursive ones.
For a more typical example, $\beta$ might be specified as ``the set of all true
arithmetical assertions of $<10^{10}$ symbols" and $\alpha$ might be the
library of all mathematical publications. To implement these ideas a function
$\I:\ \N{}^\N\times\N{}^\N\to\Re^+$ must be defined satisfying various
intuitive and technical requirements.

 In part II the random sequences, the intuitionistic free-choice sequences, and
the representatives of ``regular" Turing degrees respectively are considered as
$\alpha$. In each of these three cases a formalization of Independence
Postulate is shown to simplify radically the corresponding theories.

\subsection{\label{0.2} Brief References}

 The following remarks do not present the history of the area and concern
mainly the works directly used here. Algorithmic Information Theory originated
with the discovery of universal coding and a recursively invariant approach to
the concepts of Complexity, Information, Randomness and {\em a priori}
Probability ([Kolmogorov 65, Solomonoff]). ``Uspekhi Mat. Nauk" announced A.N.
Kolmogorov's talks on this subject in 1961 and consecutive years. Some of R.J.
Solomonoff's ideas were mentioned in [Minsky] and preprints. See also [Markov]
and [Chaitin 66,68].

 Despite the depth of the main idea, the technical expression of basic
quantities was not accurate. Many important relationships hold only with an
error such as the logarithm of Complexity. This error rate is negligible in
comparison to Complexity itself, but can exceed such differences as Mutual
Information, Deficiency of Randomness etc. The errors distorted the picture and
hindered the development of a transparent theory. The concept of Randomness was
improved in [Martin-Lof] for the case of recursive measures. It was not,
however, extendable to other important cases nor expressible in terms of a
measure-independent notion like Complexity. Very interesting studies of
Randomness were made in [Schnorr]. Some of its ideas proposed independently of
[Zvonkin, Levin 70] are related to Proposition~\ref{p4} and
Proposition~\ref{p5} below.

 The problem of giving a precise expression for Information proved to be more
difficult. The first non-trivial results were obtained by A.N. Kolmogorov and
L.A. Levin in 1967. The initial definition of $\I(x:y)$ from [Kolmogorov 65]
was asymmetric and not monotone over $y$ (with respect to projection $(y_1,y_2)
\to y_1$). In [Kolmogorov 68], [Zvonkin, Levin] this definition was
demonstrated to coincide, within logarithm of complexity, with a symmetric
expression and to be therefore approximately monotone over both arguments.

 In [Zvonkin, Levin] the Universal Measure was introduced. Its logarithm (equal
to the length of the shortest self-delimiting (or prefix) code) turned out to
be a more satisfactory complexity measure on \N\ than the original proposal
from [Kolmogorov 65]. It allowed improvement of the definitions of Randomness
([Levin 73]) and Information ([Levin 74]). The new definition of Information
was monotonic within an additive constant (rather than logarithm) and
extendable to the case of infinite sequences. This work is related to subtle
results of [Gacs 74] concerning the differences between the symmetric and
asymmetric expressions for Information. A number of results of [Kolmogorov 68,
Levin 70,73,74, Gacs 74] were also found independently in [Chaitin 75]. A
version of the present paper appeared as an MIT technical report
MIT/LCS/TR-235(1980) and some results were formulated in [Levin 70-77].

\subsection{\label{0.3} Conventions}

 $\N,\Q,\Re$ are the sets of natural, rational and real numbers. $T^+$ is the
subset $\{x:\ x\ge0\}$ of an ordered set $T\supset\{0\}$, and $\ov{T}$ is
$T\cup\{\infty\}$. The integer part of $x\in\Re$ is $\lfloor x\rfloor$. The
number $m+((m+ n) (m+ n+ 1)/2)$ is called {\em the pair} (m,n) of numbers
$m,n\in\ov\N$. This enumeration of pairs is bijective on
$\N^2\subset\ov\N{}^2$. Cantor's perfect set is represented in the form
$\Omega=\ov\N{}^\N$ which has simpler (than $\{0,1\}^\N$) expression for pairs:
$(\alpha,\beta)(i)= (\alpha(i),\beta(i))$, where $\alpha(i), \beta(i)\in\ov\N$
are the i-th terms of $\alpha$ and $\beta$. Let $S_k=\{0,1, ...,k,\infty\}^k$
and $S=\cup S_k$; The {\em length} $l(x)$ of $x\in S_k$ is $k$. If
$\alpha\in\Omega$ or $\alpha\in S_n$ and $k\le n$, then $\alpha_k\in S_k$ is
the initial $k$-segment of $\alpha$ with all terms $\alpha(i)> k$ replaced by
$\infty$; $x\subset\alpha$ or means $x=\alpha_{l(x)}, l(x)\le l(\alpha)$.

 Any property $P$ is identified with the set $\{x:P(x)\}$ and its
characteristic function. An open subset of a topological space with natural
countable basis is called {\em recursively enumerable (r.e.)} if it equals the
union of an r.e. family of basis sets. A real function $F$ on such set is
called r.e. if its {\em subgraph,} i.e., the set $\{(x,r):\ r< F(x)\}$ is r.e.
It is called {\em recursive} if $F$ and $-F$ are r.e. The symbols $\prec,
\succ$ and $\asymp$ denote inequality and equality of functions within an
additive constant; $\preceq, \succeq$ and $\simeq$ denote these relations
within a constant factor. Such operations, as $\sum f, \sup f, \min f$, etc.
are assumed taken over the values of {\em all} variables of the term $f$, not
bounded in the context.

 Conventions for section~\ref{2}. Let $\SKF$ be the space of continuous
functions $f:\ \Omega\to\Re$ with the norm $||f||=\sup|f(x)|$; Its countable
dense subset $\SKF'\subset\SKF$ consist of the functions whose ranges are
finite sets of rationals. Let $\SKJ$ be the space of lower semicontinuous
functions $\Omega\to\ov\Re$ and $\SKM$ be the set of positive linear
functionals $\SKF\to\Re$. Any $\mu\in\SKM$ represents a measure on $\Omega$ and
$\mu(f)$ is the average value of $f$. Let $\SKN$ be the set of positive linear
operators $\SKF\to\SKF$. Any $A\in\SKN$ represents a continuous random
transformation of $\Omega$ and $A(f)$ maps $\alpha$ to the average value of $f$
on the image of $\alpha$. We identify $\alpha\in\Omega$ with the measure
$\mu_\alpha(f)= f(\alpha)$, and a deterministic transformation $A:\
\Omega\to\Omega$ with the operator $f\to g$, where $g(\alpha)= f(A(\alpha))$.
Restricting $\mu\in\SKM$ to $S\subset \SKF$ gives $\mu':\ S\to\Re$, such that
$\mu'(x)=\sum\mu' (y),y\in \sigma(x)$, where $\sigma(x)=\{y:\ y\supset x,l(y)=
l(x)+1\}$ and $\mu$ is uniquely determined by such $\mu'$. Any r.e. measure is
recursive.

 Random partial processes generate sequences which may stop after a finite
number of terms. Their probability distributions satisfy only the inequality
$\mu(x)\ge\sum\mu(y), y\in\sigma(x)$. This leads to the space $\SKH$ of {\em
semimeasures,} i.e., positive, uniform, concave functionals $\mu:\ \SKF\to\Re$,
where $\mu(\SKF^+)\subset\Re^+$ and $\mu(\rho f+ g)\ge\rho\mu(f)+\mu(g)$ for
$\rho\in\Re^+$. \SKH\ is normed by $||\mu|| =|\mu(-1)|$ and ordered as
functions on $\SKF^+$. Any semimeasure equals the infinum of the measures
majorizing it. Let $\ov\mu$ be the largest measure not exceeding $\mu$. A
semimeasure $\mu$ like a measure can be extended to $-\SKJ$, as
$\mu(f)=\inf\{\mu(f'):\ f\le f '\in\SKF\}$ and to $(\ov\Re{}^+)^\Omega$, as
$\mu(f)=\sup\{\mu(f'): $ $f\ge f'\in-\SKJ\}$. If $f\in\SKJ$, then
$\mu(f)=\sup\{\mu(f '):\ f\ge f'\in\SKF\}$. Random partial transformations are
represented by positive uniform concave operators $A:\ \SKF\to\SKJ$ forming the
set $\SKS$. By $A(\mu)$ we mean $\mu'\in\SKH$ such that: $\mu'(f)= \mu(A(f))$.
All this considerations can be justified using the Hahn-Banach Theorem and
related results of functional analysis.

 The proofs are succinct and require slow reading with frequent reference to
the present section. They are, however, independent: one may skip many of them
and still understand the others.

\section{\label{1} Discrete Case}

 \subsection{\label{1.1} Complexity, Randomness and Information}

 Complexity $\K(x)$ defined in \ref{0.1} determines the length of the shortest
codes of integers $x$. A {\em prefix} algorithm $A:\ \{0,...,k-1
\}\mbox{*}\to\N$ is one defined on at most one prefix of any string.
Informally, A recognizes the end of input with no special mark and rejects any
continuation. So it has a truly $k$-ary input alphabet. The volume ${|}x|$ of
$x\in\{0,...,k-1\}^{n}$ is $\lfloor n\cdot\ln k\rfloor$.

 \begin{prp} [Coding]\label{p1} A prefix algorithm $A$ (the Huffman code)
exists generating any $x\in\N$ from some input of volume $\K(x)$; i.e.,
$\exists A,c\ \forall x\exists p\ (A(p)= x$ and $|p|=\K(x)+ c)$.

 No prefix algorithm $A$ can be better i.e., $\forall A\exists c\forall
x\forall p ((A(p)= x)\ \Rightarrow\ |p|\ge\K(x)- c)$.\end{prp}

\paragraph{Proof:} The set $G=\{(x,n):\ n>\K(x)\}$ is r.e. Then a recursive
bijection $f:\ \N\to G$ exist. Let $\lambda(x,n)=e^{- n}$ and $\mu(t)=
\sum\{\lambda(f(t')):\ t'< t\}$. Then $\mu(\infty)=\sum\lambda(x,\K(x)+ i)=\sum
e^{-\K(x)}\sum e^{- i}<1$. Let $f(t)=(x,n)$ and $.p$ be the shortest $k$-ary
fraction within $(\mu(t),\mu(t+1))$. Then $A(p)=x$.

	Vice versa, $A$ is defined on at most one prefix of any string, and is
extendable to $\{0,...,k-1 \}^\N$. The uniform measure on $\{0,...,k-1\}^\N$ is
$B(\{\alpha:\ \alpha \supset q\})=k^{- l(q)}$. If $K'(x)\asymp\min\{|q|:A(q)=
x\}$ then, obviously, $e^{- K'(x)}\prec\ B(A^{-1}(x))<1$ and $e^{- K'}
\in\ov{l_1}$. \qed

 So $\K(x)$ measures the minimal information needed to generate $x$. This makes
definition of $\I(x:y)=\K(x)+\K(y)-\K(x,y)$ more intuitive and also agrees with
Shannon's idea that the amount of information in an event equals the negative
logarithm of its probability. Obviously, $\I(x:y)\succ\ 0$, because
$(\m\otimes\m)(x,y)=\m(x)\m(y)$ is an r.e. measure and thus $\m\otimes\m\
\preceq\ \m$.

 Let us arrive at the expression $\I(x:y)$ from another point of view: the
concept of Randomness. Let $\mu$ be a recursive measure on {\N} and $t:\
\N\to\ov\Re{}^+$ be an r.e. function with average value $\mu(t)\le 1$. If
$x\in\N$ appears randomly with probability $\mu(x)$ one may expect $t(x)$ to be
not much larger then average. Then $\lfloor\ln t(x)\rfloor$ may serve as a
randomness test ($\mu$-test) for $x$.

\begin{nt}\label{n1} For any r.e. measure $\mu,\d(x/\mu)=\lfloor\ln(\m(x)/\mu
(x)) \rfloor$ is the largest (within an additive constant) $\mu$-test.\end{nt}

 So, $\d(x/\mu)\ \asymp\ |\ln\mu(x)|-\K(x)$ is, in a sense, a universal
characteristic of ``non-randomness", called the {\em Randomness Deficiency of
$x$ with respect to $\mu$}. Motivations, some history and the general
formulation of the concept of Randomness are discussed in Chapter \ref{3}.

 Let two random variables be independent and have the same distribution $\mu$.
Then their joint distribution is $\mu^2(x,y)=\mu(x)\mu(y)$. Suppose a pair
$(x,y)\in \N^2$ looks random for probability distribution $\m^2$, i.e.,
$\d((x,y)/\m^2)$ is small. This means the same as that ``(1) $x$ and $y$ look
independent and (2) each of them looks random for distribution \m". But (2) is
vacuously true, since {\em all} numbers look random for the universal
distribution \m: $\d(x/\m)\equiv 0$.

 Therefore, the smallness of $\d((x,y)/\m^2)$ means only that $x$ and $y$ could
be generated independently of each other. It is natural then to consider
$\d((x,y)/\m^2)=\I(x:y)$ as the {\em deficiency of independence.} This reminds
the theorem of classical Probabilistic Information Theory that two random
variables are independent iff they have no mutual information. The difference
is that the concepts given above are applicable to individual objects
themselves, and not only to their probability distributions (i.e., random
variables).

\subsection {\label{1.2} Conservation of Independence}

The Information $\I(x:y)$ has a remarkable invariance: It cannot be
increased by random or deterministic (recursive) processing of $x$ or
$y$. This is natural, since if $x$ contains no information about $y$
then hope is little to find out something about $y$ by processing $x$.
(Torturing an uninformed witness cannot give information about the
crime!)


 \begin{prp} [Independence Conservation] \label{p2}
 Let $f:\ \N\to\N$ be a recursive function, and $\varphi$ be an r.e. measure on
\N. Then\begin{enumerate}
 \item $\I(f(x):y)\ \prec\ \I(x:y)$,
 \item $\int\exp(\I((x,z):y))d\varphi(z)\ \preceq\
\exp(\I(x:y))$.\end{enumerate} \end{prp}

 The linear scale (instead of the logarithmic one) strengthens (2) and is more
natural with linear operator of $\int$. Proposition~\ref{p2} is an elementary
version of the Corollary of Theorem~\ref{t1} below and also implies
independence conservation in any combination of random and deterministic
(recursive) processes. This supports the Independence Postulate in \ref{0.1}.

 \paragraph{Proof:}

 \begin{lem}\label{l1} (P. Gacs): $\K(x,\K(x))\asymp\K(x)$.\end{lem}

 This elegant lemma has a short proof: let $p$ be a {\em shortest} code for
$x$. Obviously, both $x$ and $\K(x)=|p|$ are computable from $p$. Therefore,
the complexity of $(x,\K(x))$ equals $|p|= \K(x)$.

 Let the {\em universal conditional measure} be a largest within a constant
factor r.e. function $\m(/):\ \ \N^2\to\Re^+$ such that
$\sup_y\sum\m(x/y)<\infty$, and $\K(x/y)=-\lfloor\ln \m(x/y)\rfloor$.

\begin{lem}\label{l2} $\K(x,y)\ \asymp\ \K(x)+\K(y/(x,\K(x)))$.\end{lem}

 Let $\m_\infty(y/x,n)=e^n\m(x,y)$. A non-decreasing by k, recursive sequence
$\m_k(y/x,n): A_k\to\Q^+$ exists, such that $\m_\infty= \sup\m_k$, where
$A_k\subset\N^3$ are finite. Let $k(x,n)=\sup\{k:\ \sum m_k(y/x,n)\le1\}$, and
$\ov\m(y/x,n)=\m_{k(x,n)}(y/x,n)$. Obviously $\forall x,n\ \sum\ov
\m(y/x,n)\le1$ (thus $\ov\m(/)\preceq\m(/)$) and $\forall x,n$ if $\sum\m(x,y)
\le e^{- n}$, then $\m_\infty(y/x,n)=\ov \m(y/x,n)$. Therefore $\forall x,n$ if
$\sum\m(x,y)\le e^{- n}$, (and thus $\m(x) \preceq e^{- n}$ or $n\succ\K(x)$)
then $\m(y/x,n)\succeq\ov\m(y/x,n)=\m_\infty (y/x,n)=e^n\m(x,y)$. Thus
$\K(y/x,\K(x))$ $\prec\K(x,y)-\K(x)$.

	It remains to prove that $\K(y/(x,\K(x)))\succ\K(x,y)-\K(x)\asymp
\K(x,y)-\K(x,\K(x))$.\hspace{2pc} This follows from $\K(x,y)\prec
\K(y,x,\K(x))$, $\K(x)\asymp\K(x,\K(x))$ and $\K(y,t)\prec\K(t)+\K(y/t)$.
\hspace{1pc} The latter inequality holds since $\m'(y,t)=\m(t)\m(y/t)$ is
obviously an r.e. measure and then $\m'(y,t)\preceq\m(y,t)$.

{Now, \ \ $\K(x,y,z)$ $\prec$ $\K(x,\K(X))+\K(y/(x,\K(x)))+ \K(z/(x,\K(x)))$
$\asymp$ $\K(x,y)+\K(x,z)-\K(x)$,}\\
 {since $\K(y,z/t)\ \prec\ \K(y/t)+\K(z/t)$. \hfil Therefore $\I((z,x):y)\
 \succ\ \I(x:y)$ and (1) follows by noting that}\\
 $\I(z:y)\ \asymp\ \I((z,x):y)$, for $x= A(z)$ since $z$ and $(z,A(z))$ are
computable from each other.

To prove (2) we need to show that: $\m(x,y)/(\m(x)\m(y))\ \succeq\
\int(\m(x,y,z)/(\m(y)\m(x,z)))d\varphi(z)$, or\\
 \mbox {$\int(\m(x,y,z)/\m(x,z))d\m(z)\preceq\ \m(x,y)/\m(x)$, since
$\m(z)\succeq\varphi(z)$.} Rewrite it:
 $\sum_z\m(z)\m(x,y,z)/\m(x,z)\preceq\ \m(x,y)/\m(x)$; or
 $\sum_z\m(z)\m(x)\m(x,y,z)/\m(x,z)\ \preceq\m(x,y)$.\\
 The latter is obvious since: $\m(z)\m(x)\ \preceq\ \m(x,z)$ and
$\sum\m(x,y,z)\ \preceq\ \m(x,y)$. \qed

\subsection {\label{1.3} Time of Computation}

 The speed of generating various r.e. sets is, as a rule, ignored in this work.
Now we touch this question briefly. Let $t_{A(p)}$ mean the running time of
$A(p)$. If $A$ is the optimal algorithm from Proposition~\ref{p1} and $R(p)$ is
$(A(p),|p|)$ then $\forall x\exists p:\ R(p)=(x,\K(x))$. Exhaustive search for
such $p$ takes exponential time, even when $R(p)$ is fast. Let us give a
fastest algorithm (storage modification machine) finding $p$.

 Let $Kt_B(x/y)=\min\{(|p|+\ln t_{B(p,y)}):\ B(p,y)= x\}$, where $p$ is a
string without termination mark: the algorithm $B$ receives, upon request, the
digits of $p$ in order until $p$ is ended; in case of further requests $B$ gets
no reply and gives no output. $Kt_B(x)= Kt_B(x/0)$. Analogously to
Proposition~\ref{p1}, an optimal $B$ exists such that $Kt_B$ is minimal within
an additive constant, and $Kt_B$ is denoted by $Kt$. There exists an algorithm
$G(n,y)$ generating the list $\{x:\ Kt(x/y)= n\}$ in time $e^n$; and within a
constant factor, $Kt$ is a minimal function with this property. (The
asymptotically minimal one is $kt(x)=\min\{Kt(a):\ x\in a\subset \N\}$). Let
$R$ be some function computable in polynomial time. A problem of finding $q\in
R^{-1}(s)$ (when it exists) is called a search problem. An NP-problem is to
find whether a (short) such $q$ exists. Wlog, $t_R$ can be made linear by
``padding" $q$ with zeros. Searching through all $q$ in the order of increasing
$Kt(q/s)$ ({\em rather than of} $|q|$) gives a fastest (within a constant
factor) algorithm for solving any search problem (see [Levin 73a]; related idea
were also expressed by L. Adleman).

Functions like $Kt$ are of a particular interest for the case of
randomized algorithms. For $f:\N\to\ov\Re{}^+$ let $c(f)$ be the
expected value of $1/(t_{B(\alpha)}+ f(B(\alpha)))$, where $\alpha$
is the random variable and, like above, $B$ is the optimal algorithm.
 Let $C(f)=\lfloor-\ln c(f)\rfloor$. For $F\subset\N$, $C(F)$
 means $C(f)$, where $f(x)=0$, if $x\in F$, else $f(x)=\infty$.
 The above algorithm $G(n,y)$, finds numbers $x\in F\subset\N$ using
random $y=\alpha$ in time $<e^n$ with probability $p>1-e^{-k}$, where
$\ln k$ equals {{$n-C(F)-2\ln(C(F))$.}}
 For any other algorithm, $p<e^{n- C(F)}$ (otherwise $C(f)$ could
be improved). Thus $C(F)$ determines the time needed to ``hit" $F$.
A function $f$, with range other then just $\{\infty,0\}$, can be
interpreted as a ``price" (e.g. time) needed to establish $x\in F= f
^{-1}(\Re^+)$.

 Everything is analogous for $C(f/y)=-\lfloor\ln \int d\alpha/(t_{B(\alpha,y)}
+ f(B(\alpha,y),y))\rfloor$. C. Bennett proposed an interesting electrical
interpretation of $C(F/y)$. Let us take the circuit of wires connected in
parallel, each of which corresponds to a particular finite $\alpha$ such that
$B(\alpha,y)\in F$. The length and the probability of a computation gives the
length and the section of the wire. Then $C(F/y)$ is the logarithm of the
circuit's resistance.

 A number of search (NP) problems are known, which are easy for probabilistic
algorithms, but seem hard for deterministic ones. E.g. constructing
``incompressible" words $x$, of high $K_c(x)$, where $c$ is a constant and
$K_c(x)$ (computable in polynomial time) is the minimal length of $p$, with
$Kt(x/p)+ Kt(p/x) <c\log |x|$. The complexity of a search problem $R$ for
probabilistic algorithms is characterized by $C(f_s/s)$, where
$f_s(q)=t_{R(q)}$, if $R(q)=s$, and $f_s(q)=\infty$ otherwise. The relationship
of this complexity with $|s|$ is a ``randomized" version of the P=NP problem.
But its relationship with the ``complexity of obtaining s" looks even more
interesting. More accurately: how does $C(\{s:\ \infty> C(f_s/s)> n\})$ grow
with $n$, polynomially or logarithmically? Short $s$ may exist for which it is
very difficult to find $q\in R^{-1}(s)$, but to find such $s$ may be even more
difficult.

\section {\label{2} Continuous Case}

	\subsection {\label{2.1} Universal Semimeasure}

	Let us extend now {\m} to the case of $\Omega$.

	\begin{prp}\label{p3} There exists a largest within a constant factor
(universal) r.e. semimeasure:

$\exists\M\forall\mu\ \exists c\forall f> 0:\ \ \ \mu(f)\ \le\ c\cdot\M(f)$.
\end{prp}

	\paragraph{Proof:} Let $(H,<)$ be an ordered set with a monotone
operation $a:\ H^2\to H$ (``averaging") and a family $Y:\N\to H$, containing
the smallest element 0, closed under $a$ and such that ``$<, a$" are r.e. on
the indexes. Let the supremums of all directed sets $Y(A)$ with r.e. $A$ exist
and are called the r.e. elements. Then $(H,Y,<,a$) will be called a {\em
numbered convex body.}

\begin{nt}\label{n2} Any numbered convex body $(H,Y,<,a$) has a universal r.e.
element, i.e. one largest in any weaker then ``$<$" order ``$\preceq$" such
that $x\preceq a(x,y)\succeq y$.\end{nt}

 This element is the supremum of $\cup a_k(0)$, where $a_0(x)=\{x\}$,
$a_{k+1}(x)=\cup\{a_k(a(r,x)):\ r\in t_k\}$ and $t_k$ is the k-th enumerable
directed subfamily of $Y$.

 Proposition~\ref{p3} is a special case, where $H=\SKH$, $a(\mu,\varphi)=
(\mu+\varphi)/2$ and $Y$ be the family of ``elementary" semimeasures in a
natural numeration. An elementary semimeasure is the minimal one satisfying a
finite set of inequalities $\mu(f)>r$, where $f\in\SKF', r\in\Q$. \qed

 This {\M} is the central technical concept of the work. Being largest, it
determines the broadest class of sets $A\subset\Omega$ of positive probability.
In statistics one tries, given $\alpha$, to get a probability distribution
$\mu$ with respect to which $\alpha$ may be reasonably considered ``random".
I.e. $\mu(A)>>0$ for ``standard" properties $A$ satisfied by $\alpha$. But this
assertion is the weakest with $\mu =\M$. So, {\M} can be taken {\em a priori,}
before studying what the properties of $\alpha$ really are. In other words. If
$\alpha$ occurs randomly with probability distribution $\mu$ then it has
properties $A$, for which $\mu(\neg A)=0$. The class of such properties A is
narrowest for $\mu = \M$, and they can be predicted {\em a priori} before
finding out what $\mu$ really is. This justifies using {\M} as {\em a priori}
probability.

 The distribution {\M} is suitable for other applications, as {\em a priori}
probability (e.g. for inductive inference in accordance with the ideas of
[Solomonoff]), but this questions are not considered here. Let us note that
$\mu(\alpha_n)\ \simeq\ \M(\alpha_n)$ for any r.e. semimeasure $\mu$ and
$\mu$-almost all $\alpha$. This property of $\alpha$ can be used (see
Proposition~\ref{p5}) as a definition of the concept of a $\mu$-random
sequence.

\subsection {\label{2.2} Randomness and Information; Conservation Laws}

 Now the second half of the Conventions is very essential. Let us extend to
$\Omega$ the concept of randomness tests considered in 1.1. For any set $A$ of
$\mu$-measure 0 there is a lower semicontinuous function $t\in\SKJ^+$, with
average value $\mu(t)\le1$ and $t(A)=\{\infty\}$. Only for recursive $\mu$,
r.e. tests $t$ are natural to consider. For a general case let $t^\rho
\in\SKJ^+$ be the function whose subgraph is enumerated by $\rho\in\Omega$. In
 Definition~\ref{d2} we will average this over all $\rho$ generated
``arbitrarily" i.e. randomly with universal distribution \M. We eliminate the
case of $\mu(t ^\rho)>1$ defining $t_\mu^\rho=t^\rho$, if $\mu(t^\rho)\le1$,
and $t_\mu^\rho=0$ otherwise. To deal with random transformations which may
turn an individual sequence $\alpha$ into a measure $\varphi$, we will extend
tests $t_\mu^\rho(\alpha)$ to distributions as $\varphi(t_\mu^\rho)$. Let
$(^\varphi _\mu)\in\ov\Re{}^\Omega$ maps $\rho$ to $\varphi(t_\mu^\rho)$.

 \begin{dfn}\label{d2} $\D(\varphi/\mu)=\lfloor\ln\M(^\varphi_\mu)\rfloor$ is
the Deficiency of Randomness of $\varphi$ with respect to $\mu$.\end{dfn}

 Most important is the special case $\D(\alpha/\mu)$ when $\varphi$ is
concentrated in a single point $\alpha$. It is a generalization of $\d(x/\mu)$
from Note~\ref{n1}, since its ``average" is $\D(\mu/\mu)\le0$. For motivations
of the notion of Randomness one may look in sections~\ref{3} and \ref{1.1}. Now
the central fact comes:

\begin{thr} [Randomness conservation] \label{t1} $\D(A(\mu_1)/ A(\mu_2))\prec
\D(\mu_1/\mu_2)$, where $A\in\SKS$ is r.e., $\mu_1,\mu_2\in\SKH$.\end{thr}

 \paragraph{Proof:} Let $A'(\rho)$ be a sequence, enumerating the subgraph of
$A(t^\rho)\in\SKJ^+$.

\noindent Then
$\exp\D(A\varphi/A\mu)=\M(^{A\varphi}_{A\mu})=\M(A'(^\varphi_\mu)) =(\M\circ
A')(^\varphi_\mu)\preceq \M(^\varphi_\mu)$, since {\M} is universal. \qed

 So, $\D(\alpha/\mu)$ is invariant (within a constant) with respect to all r.e.
operators preserving $\mu$.

 \begin{dfn}\label{d3} $\I(\alpha:\beta)= \D((\alpha,\beta)/ \M\otimes \M)$
 is called the {\em amount of information in $\alpha$ about $\beta$} or the
{\em deficiency of their independence.} Here $\alpha,\beta$ may be sequences or
semimeasures.\end{dfn}

 \begin{cor} (A generalization of Proposition~\ref{p2}) Suppose
$\alpha,\beta\in\Omega$. Then $\I(A(\alpha):\beta)\prec
\I(\alpha:\beta)$.\end{cor}

 The proof follows by noting that $A(\M)\preceq \M$ and {\D} is monotone.

 This justifies the Independence Postulate, from Introduction, since one can
usually ``explain" known physical processes reducing them to simpler ones in
combination with recursive and random transformations (considered in the above
Corollary). The Universe, on the whole, is also assumed to evolve according to
the (recursive) equations of physics from a state of random movement of hot
plasma (additional randomness appears in the observation processes). Of course,
not being a mathematical assertion (the physical world is not chosen
mathematically), the Independence Postulate (like, e.g., Church's thesis)
cannot be proven.

\subsection{\label{2.3} Complete Sequences}

 Any r.e. measure is recursive, i.e., computable with any accuracy, in contrast
to semimeasure $\M$ for which any r.e. lower bound of $(\max f(x))/\M(f)$ is
bounded by a constant. But it may be known about some $\alpha\in\Omega $ that
on its segments {\M} agrees with some r.e. {\em measure} $\mu$ within a
constant factor. Then, computing $\mu$ gives $\M(\alpha_n)$ within a constant.
Such $\alpha$ is called {\em complete,} denoted $\alpha\in C$ or
$C(\alpha)\iff\ \exists\mu \sup(\M(\alpha_n)/\mu (\alpha_n))<\infty$. Its
segments contain all information needed to compute their complexity
$-\lfloor\ln\M(\alpha_n) \rfloor$. According to Proposition~\ref{p4}, $C$ is
very wide. By virtue of its item 2, any sequence $\alpha$ satisfying the
Independence Postulate has a completion $(\alpha,\tau)\in C$, satisfying this
Postulate as well.

\begin{prp}\label{p4}\begin{enumerate}
 \item For any recursive measure $\mu$ and total recursive operator $A$,
$\mu(C)=\mu(\Omega)$ and $A(C)\subset C$.
 \item Let $\beta$ be a sequence to which a universal r.e. set is Turing
reducible and $\alpha$ be independent of $\beta$. Then $\tau\in \N^\N $ exists
such that $(\alpha,\tau)$ is complete and independent of $\beta$.
 \end{enumerate}\end{prp}

 This $\alpha$ comes from $(\alpha,\tau)$ by a partial recursive (but not
total) projection operator. So incompletable partial operators can lead out of
$C$. In Chapter \ref{4} an axiom is used which means intuitively that every
``physical" sequence is a projection of a complete one.

\paragraph{Proof:} Let $r<1$ and $\delta_\mu(\alpha)=r\,\ln\sup
(\M(\alpha_n)/\mu(\alpha_n))<\infty$ and $\mu'(x)=\mu\{ \alpha:A(\alpha)\supset
x\}$. Then $\mu$'\ is also an r.e. measure, and $\delta_{\mu'}(A(\alpha))$ is a
$\mu$-test. Then by virtue of Proposition~\ref{p5},
$\delta_{\mu'}(A(\alpha))<\infty$ and $A(\alpha)\in C$. Obviously $\mu(C)= 1$
because $\delta_\mu(\alpha)$ is a randomness test.

	It remains to prove (2). In section~3.2 of [Zvonkin, Levin] it is shown
that {\M} (like any other r.e. semimeasure) is generated by a partial recursive
operator $A$ from a recursive measure $\mu:\ \M(x)=\mu \{\alpha:
A(\alpha)\supset x\}$. Let $A'(\alpha)=(A(\alpha),\tau_{A(\alpha)})$, where
$\tau_{A(\alpha)}$ is the sequence of values of the time of $A(\alpha)$ terms
calculation. The operator A'\ is total and, hence, $\mu'(x)=\mu\{
\alpha:A'(\alpha)\supset x\}$\ is a recursive measure. {\M} is generated from
$\mu$'\ by the projector $(\alpha,\tau)\to\alpha$. And $\mu'\{(\alpha,\tau) :\
\I((\alpha,\tau):\beta)= \infty\}\ =0$, by the Note~\ref{n3}. Also
$\mu'(\Omega-C)= 0$. Therefore, $\M\{\alpha:\ \forall \tau\in\Omega\
(((\alpha,\tau)\mathrel{\not\in} C)\vee\I((\alpha,\tau):\beta)= \infty)\} =0$.
By the Note~\ref{n3}, for any set $A$ such that $\ov\M(A)=0$, a sequence
$\beta$'\ exists on which all elements of $A$ depend. The same is true for any
sequence to which $\beta'$ is reducible. Using reducibility to $\beta$ of the
universal r.e. set, one can routinely check that the necessary $\beta$'\ is
computable with respect to $\beta$. Thus $\beta$ depend on all sequences
$\alpha$ not completable to a complete, independent of $\beta$ sequence
$(\alpha,\tau)$. \qed

\newpage\centerline{\bf\large II. Applications}\vspace{1pc}

	\section{\label{3} Foundations of Probability Theory}
 \subsection{\label{3.1} Foundational Difficulties (A historical digression)}

 Hilbert's sixth problem suggests ``{\bf To treat \rm in the same manner} (as
geometry), {\bf by means of axioms, those physical sciences, in which
mathematics plays an important part; in the first rank are the Theory of
Probabilities \rm and mechanics."} (see [Hilbert]). The Probability Theory is
considered there as a physical science. But its physical nature was almost
forgotten after the analytical explosion followed A.N. Kolmogorov's book
[Kolmogorov 33], where Probability Theory obtained a powerful mathematical
foundation. And some difficulties in relation of this mathematical apparatus to
the probabilistic natural phenomena were left aside. Kolmogorov noted this in
the foreword to the second Russian edition of the book, where he refers to
[Kolmogorov 65], [Zvonkin, Levin] for his new approach. The well-known previous
attempts to overcome these difficulties [vMises, Church] turned out to be
imperfect [Ville].

 The difficulties lie in the gap between intuitive probabilistic ideas and
those methods which are justifiable theoretically. Probabilistic considerations
start with an assumption that a sequence $x$ is generated randomly with
probability distribution $\mu$. This $\mu$ is discovered or hypothesized e.g.
by analogy with other processes and statistical data about them, considerations
of symmetry, etc. Then, according to the naive ideas, those properties of $x$
are indicated as probabilistic laws whose $\mu$-probability is 1
(approximately, in the finite case). E.g. the Law of Large Numbers plays an
important role when $x = x_1,...,x_n$ are independent and identically
distributed trials (i.e., $\mu(x_1,...,x_n)=\mu'(x_1) \mu'(x_2)...\mu'(x_n)$).
For each set $B$ it predicts $x\in LBN$, which means that the frequency of $i:\
x_i\in B$ is close to the probability $\mu'(B)$ (and $\mu(LBN)\approx 1$). In
general, $x$ is predicted to have the properties whose probabilities are close
to 1.

 The problem is that {\em jointly the properties of probability 1 have
probability 0 !} One cannot guarantee all of them simultaneously, but should
choose one or a few. Thus, the outcome should not be expected to withstand
statistical tests chosen afterwards. So classical theory provides no rigorous
basis to doubt the honesty of the lottery director whose son won the main prize
in ten consecutive years, if this is discovered ``post factum"! One cannot
criticize an election when the share of votes for the ruling party in a series
of consecutive years formed a sequence $0.99k_i$, even if $k_i$ turn out to be
the decimal digits of $\pi$ ! Of course, one can select few ``standard laws"
and presume them always to be chosen. {\em However, the classical Probability
Theory has no ideas for selection of such standard laws.} Besides, this would
not justify applying Probability Theory to events preceding such a
standardization (e.g., in cosmology, history, geology, etc.).

 The Kolmogorov's idea for solving this paradox is to select those properties
of probability close to 1, which are ``simply expressible". The objects not
satisfying such a property form a simple set of small measure and
correspondingly small cardinality. Then any such object is simple itself, being
specifiable by its number (smaller then the set's cardinality) with the simple
set's description. This allows substitution of many simple properties by a
single one: ``not to be a simple object". Kolmogorov's Algorithmic Information
Theory was a surprising discovery, provided a rigorous basis for the obscure
notion of simplicity. In the infinite case the corresponding property is ``to
be random with respect to distribution $\mu$". Then only this property is
expected from the objects occurring randomly with distribution $\mu$. This
property is of $\mu$-measure 1 and implies all other ``good" properties of
$\mu$-measure 1. Attempts to introduce such universal concepts of Randomness
were undertaken previously in [vMises, Church] for Bernoulli distributions.
However, it was found [Ville] that even such standard properties as the Law of
Iterated Logarithm do not follow from their concept of being a ``Collective".

\subsection{\label{3.2} The Laws of Randomness and Independence}

 Any set $A$ is of $\mu$-measure 0 iff there is a lower semicontinuous function
$t\in\SKJ^+$, with average value $\mu(t)\le 1$ and $t(A)= \{\infty\}$. For a
typical result $\alpha$ of a $\mu$-distributed random process, $t(\alpha)$
should not exceed much the average and the probability of large deviations is
small. This justifies the following modification of a definition from
[Martin-Lof].

 \begin{dfn}\label{d4} A randomness test with respect to a recursive measure
$\mu$ (or a $\mu$-test) is a function $\lfloor\ln t(\alpha)\rfloor$, where
$t\in \SKJ^+$ is r.e. and $\mu(t)\le 1$.\end{dfn}

  Definition~\ref{d4} is a formalization of the concept of a ``good" law of
Probability Theory. The degree $\delta(\alpha)$ of deviation from such a law is
absolute when $\alpha$ fails the test i.e. $\delta (\alpha)= \infty$, the
probability of which is 0. The deviations can be effectively discovered since
$\delta$ is r.e. The logarithmic scale is chosen for convenience.


 \begin{prp}\label{p5}
 The following properties of $\alpha\in\Omega$ are equivalent for any recursive
measure $\mu$ and true for $\mu=\M:$\begin{enumerate}
 \item $\D(\alpha/\mu)<\infty$. (see Definition~\ref{d2}.)
 \item $\sup (\M(\alpha_n)/\mu(\alpha_n))<\infty$.
 \item For any randomness test $\delta:\ \delta (\alpha)<\infty$.
 \end{enumerate} \end{prp}

 Item 2 means that $\mu$ on segments of $\alpha$ is not much smaller than the
{\em a priori} probability \M. (i.e., the assumption that $\alpha$ has occurred
randomly with probability distribution $\mu$ is at least as consistent with
reality as the {\em a priori} idea about occurring it with distribution \M).
This property is of $\mu$-probability 1 and implies all other effective
probabilistic laws.

 {\bf Proof:} For any recursive measure $\mu$ and $r<1$, it is easy to see that
$r\ln t$ is a $\mu$-test itself, where $t\simeq\sup\{\M (\alpha_n)/\mu
(\alpha_n)\}$ and thus (3)$\Rightarrow$ (2). For any semimeasure $\mu$, (1)
$\Rightarrow$(3) since the sequence $\rho$ enumerating the subgraph of
$\exp{\delta}$ is r.e. (thus $\M(\{\rho\})>0$) and $\mu(\exp\delta)\le1$. For
$\mu=\M$, (2) is obvious and it remains to prove (2)$\Rightarrow$(1) for any
r.e. semimeasure $\mu$. Let $\mu'_\rho$ is an r.e. family of r.e., with respect
to $\rho$, normalized semimeasures, such that $\mu'_\rho=t^\rho\cdot\mu $ if
$\mu(t^\rho)\le 1$. Let $\tau_\mu^f(\rho)=\mu'_\rho(f)$. Then
$\mu''(f)=\M(\tau_\mu^f)$ is obviously an r.e. semimeasure and
$\M(f)\succeq\mu''(f) \ge\exp\D((\mu \cdot f)/\mu)$. Let $\varphi_x(f)=
\min\{f(\alpha):\ x\subset\alpha\}$. Let $g_x(\alpha)=1$, if $x\subset\alpha$
and $0$ otherwise. Since $\mu\cdot g_x\ge\mu(x)\cdot\varphi_x$, we have: $\M
(x)/\mu(x) \succeq\exp \D(\varphi_x/\mu)$ and
$\D(\alpha/\mu)=\sup\{\D(\varphi_x/\mu ):\ x\subset\alpha\}
\prec\sup\{\M(\alpha_n)/\mu(\alpha_n)\}$. \qed

 Now we are ready to formulate the first of the two distinguished probabilistic
laws:

 \begin{law} [of Randomness] Let $\alpha\in\Omega$ be taken randomly with a
probability distribution $\mu$. Then $\D(\alpha/\mu)<<\infty$. \end{law}

 This Law was shown to imply all r.e. probabilistic laws. What about the other
ones? This is interesting for clarifying the relation between the algorithmic
and classical approaches to Probability Theory. Let us give an important
example of non-recursive laws.

 \begin{law} [of Independence]
 Let $\beta\in\Omega$ be chosen. A random $\alpha\in\Omega$ must be independent
of $\beta$, i.e., $\I(\alpha:\beta)<< \infty$.\end{law}

 \begin{nt}\label{n3} The assertions: ``$\ov\mu(A)= 0$ {\em for all r.e.
semimeasures} $\mu$" and {\em``there exist $\beta$, such that $\I(\alpha:\beta)
=\infty$ for all} $\alpha\in A$" are equivalent for any $A\subset\Omega$, as it
will follow from Lemma~\ref{l3}.\end{nt}

\subsection{\label{3.3}
 Covering the Classical Formulation of Probability Theory}

 The Law of Independence (as well as of Randomness) is violated only with
probability 0, and thus it is a law of Probability Theory in the customary
``classical" sense. This law varies only with the parameter $\beta$, and in its
formulation ($\I(\alpha:\beta)<\infty$) the probability $\mu$ is not mentioned
at all.

 The Independence Postulate in \ref{0.1} extends this law from the usual random
processes to any physically realizable ones. This suggests to bring this law
outside the bounds of Probability Theory and to consider other probabilistic
laws only for sequences which satisfy the Law of Independence. This turns out
to reduce any other probabilistic law (recursive or not) to the ``Law of
Randomness". Let $\I_\beta=\{\alpha:\ \I(\alpha:\beta)= \infty\}$,
$\D_\mu=\{\alpha:\ \sup(\M(\alpha_n)/\mu(\alpha_n))=\infty\}$ (= $\{\alpha:\
\D(\alpha/\mu)=\infty\}$ for an r.e. $\mu$).

 \begin{thr}\label{t2} Let $A\subset\Omega$ and $\mu (A)= 0$, then such $\beta$
exists that $A\subset\D_\mu\cup\I_\beta$, i.e., any probabilistic law follows
from the Laws of Randomness and Independence. \end{thr}

\paragraph{Proof:} Since $\mu(A)=0$ and $\alpha\mathrel{\not\in} \D_\mu\
\Rightarrow\ \M(\alpha_n)\simeq\mu(\alpha_n)$, we have
$\ov\M(A\setminus\D_\mu)=0$. \ \ Then there exist $\tau\in\SKJ^+$ such that,
$\ov\M(\tau)<\infty$ and $\tau(A')=\{\infty\}$, and the Theorem follows from

\begin{lem}\label{l3} For each borel $\tau\in(\ov\Re{}^+)^\Omega:\ \ \
\ov\M(\tau)<\infty$ iff $\beta\in\Omega$ exists such that
$\I(\alpha:\beta)\succ\ln\tau(\alpha)$.\end{lem}

``If" follows from Theorem~\ref{t1} (for $A:\ \beta\to\M\otimes\beta$) and item
1 of Proposition~\ref{p5}, since for any $\beta$: \\ $\ln
\ov\mu(\exp\I(\alpha:\beta))\ \preceq\ \D(\ov\mu\otimes\beta/\M\otimes\M)$
$\preceq\ \D(\M\otimes\beta/\M\otimes\M)\ \preceq\ \D(\beta/\M)<\infty$.

	Other side: $\M(\tau)<\infty$, then $\beta\in(\Q^+)^\N$ exists such
that $\M(f)>1\ \Rightarrow\ \beta(f*)=0$ (for a natural effective enumeration
$f\to f*$ of $\SKF^+$), and $\sigma=\sum\beta(i)<\infty,\ \tau<\sum
f\cdot\beta(f*)$. Let $L$ be the uniform measure on $[0,\sigma]$ and,
$s(\rho,\beta)$ be the largest integer for which $\sum\{\beta(i):\ i<
s(\rho,\beta)\} <\rho\in[0,\sigma]$. Let $A(\rho ,f,k)$ be a sequence
enumerating the subgraph of $\chi$, where $\chi(\alpha,\beta)=e^kf(\alpha)$, if
$f*= s(\rho,\beta)$, otherwise $\chi(\alpha,\beta)= 0$. Obviously $\M^2(\chi)\
\preceq\ \M(f)\cdot e^k\cdot\m(f*/\rho)$ and thus, for some constant $c$, if
$\K(f*/\rho)\ge k+ c,\ \M(f)\le1$, then
$(^{(\alpha,\beta)}_{M^2})(A(\rho,f,k))= \chi(\alpha,\beta)$. Like in
Lemma~\ref{l1}, $\m((f*,\K(f*/\rho))/ \rho)\ \simeq$ $e^{-\K(f*/\rho)}$ and
then: $\exp\I(\alpha,\beta)=\M(_{M^2}^{(\alpha,\beta)})\ \succeq\ \sum
f(\alpha)\cdot L\{\rho:\ s(\rho,\beta)= f*\}$. And $\I(\alpha,\beta)\ \succ\
\ln\sum (f(\alpha)\cdot\beta(f*)) =\tau(\alpha)$, since $L\{\rho:\
s(\rho,\beta)= f*\}=\beta(f*)$. \qed

\section{\label{4} Intuitionistic Mathematics}

	\subsection{\label{4.1} A Digression}

 The second order theories (permitting quantification over functions or
sequences) are much more complicated logically than the first order ones. Some
mathematicians (like intuitionists) considered these complications dangerous in
terms of possible paradoxes. In particular, they assume sequences to be formed
by sequential ``free choices" thus resulting from ``physical" (or mental)
events rather then from logical definitions. Therefore, applicability of usual
logical operations to them is not {\em a priori} obvious when these operations
have no ``physical" analogies. For example, classical universal quantification
assumes the unrealistic ability to scan all conceivable sequences. The hope was
to get a less ``suspicious" mathematics restricting the logical procedures and
postulates to only those closer related with ``physical intuition".

 However, obscurity of notions like ``free choice" and of the physical
intuition makes it difficult to choose formal principles reflecting adequately
the nature of physically generated sequences. A result is a great variety of
intuitionistic principles and theories that strengthen, weaken or contradict
each other. These theories are often so strong that the relation of their
principles to physical intuition stops to be obvious. In fact, they are often
equiconsistent to the corresponding classical theories. On the other hand, they
are too weak leaving independent many other principles of intuitionistic
reasoning. This provides room for creativity in extending these theories, but
eliminates the hope to get a ``canonical" theory with some kind of
completeness.

 To deal with these difficulties an axiom schema is introduced below
formalizing Independence Postulate from Introduction. The intuitionistic {\em
second} order arithmetic obtained is equiconsistent to (is a {\em conservative}
extension of) the classical {\em first} order arithmetic, formulated without
disjunction and existential quantifier. So it's new principles are ``purely
logical", i.e., imply no new facts of classical Number Theory. On the other
hand, it is {\em complete} i.e. is a {\em maximal} conservative extension. The
Independence Postulate brings these ``virtues" by excluding sequences with
unbounded information about the validity of mathematical statements. It is
natural to attribute the usual troubles of second order theories to such fancy
``logical" sequences which, in fact, can not physically exist.

\subsection {\label{4.2} The Preliminary Calculus A}

 Our theory AI will be constructed in Section~\ref{4.3} by extending the basic
calculus A, described below. The language of $A$ is the second order
arithmetic. It contains the first-order arithmetic (see [Kleene 67]
section~38), a countable list of second order variables (denoting sequences or
functions of natural numbers), terms $\alpha(t)$ and formulas $\forall\alpha F$
and $\exists\alpha F$, for all terms $t$, formulas $F$, and second order
variables $\alpha$. A formula is called {\em absolute} if it does not contain
$\exists,\ \vee$ and quantification over {\em second order} variables. Absolute
formulas have identical meaning and equivalent provability in intuitionistic
and classical theories. Pairs (and tuples) of integers, sequences and terms are
defined the same way as in Notation. Abbreviations $s=\alpha(n,\downarrow)$
mean $\exists k(\alpha(n,k)= s+ 1 \& (\forall k'< k\ \alpha(n,k') = 0))$.

 The postulates of $A$ consist of the ones of first order arithmetic (see
[Kleene 67], p.387, List of Postulates, Schema 8 is taken in the intuitionistic
version 8') and three second order postulates:

\begin{description}
 \item [Schema of Choice] $(\forall n(\neg A\Rightarrow\exists k\ B(k)))\
 \Rightarrow\ \exists\alpha\forall n(\neg A\ \Rightarrow\
 B(\alpha(n,\downarrow)))$ \hfil ({4.2.1})\hspace{1pc}
 \item [Markov Principle] $(\neg\forall n\ \alpha(n)=0)\ \Rightarrow\
 \exists n\ \ \alpha(n)\ne 0$ \hfil ({4.2.2})\hspace{1pc}
 \item [Axiom of Countability] $\exists\alpha\forall\beta\exists k\forall n\
 \beta (n)=\alpha(k,n,\downarrow)$ \hfil ({4.2.3})\hspace{1pc}\end{description}

 The axioms of $A$ are not new and need no detailed discussion. Still in any
complete (in a sense of Theorem~\ref{t3}) theory these axioms must be provable
or refutable or equivalent to some undecidable absolute statements of Number
Theory. The last two possibilities seem to be less natural. It is known that (
(4.2.1) - (4.2.3)) are inconsistent with the principle of continuity.

 Of course, the calculus $A$ is still too weak. Nevertheless:

 \begin{prp}\label{p6} For any formula $F$ an absolute $P$ exists such that $A\
\vdash\ F\iff\forall\alpha\exists\beta\ P$.\end{prp}

\paragraph{Proof:} Axioms of $A$ allow a construction analogous to Kleene's
recursive realizability, using the universal sequence $\alpha$ from axiom
(4.2.3). Namely, a formula $P_F(x,\alpha)$, meaning ``a number $x$ realizes a
formula $F$ with respect to a sequence $\alpha$" can be defined as in Kleene's
{\em Introduction to Metamathematics,} Chapter 2, except that recursiveness of
all functions used is considered with respect to $\alpha$. Then any formula $F$
is equivalent to the existence of a realization of $F$ with respect to a
universal $\alpha$. It is easy (though bulky) to check that $A$ contains all
the axioms necessary for formalizing these arguments, i.e., the deduction of
$F\iff\forall\beta\exists\alpha\ P_F(\alpha(0),\ (\alpha,\beta)\ )$. \qed

\subsection{\label{4.3}
	The Calculus AI;\ its Relative Consistency and Completeness}

 Let $P(n)$ be an absolute formula with a single free variable $n$. A finite
binary sequence $p$ is {\em compatible} with $P$ (denoted $p\subset P$), if
$\forall n\le l(p):\ (p(n)= 0\ \iff\ P(n))$. The abbreviation $\I(\alpha:P)$
means $\sup\{\I(\alpha:p):\ p\subset P\}$. For a given $P$, the statement
$\I(\alpha:P)\le c$ can be easily expressed by an absolute formula with free
variables $\alpha$, $c$. This is used in the following axiom schema with a
parameter P:

\begin{Postulate} [Independence]
 $\exists c\ \I(\alpha:P)\le c$ \hfil (IP).\end{Postulate}

 The property of completeness (defined in \ref{2.3}) is expressible by an
absolute formula $C(\alpha)$. The last axiom of AI asserts implementability of
sequences completion mentioned in Proposition~\ref{p4}: $\exists\tau
C(\alpha,\tau)$. The double negations of this axiom and of (IP) would be
sufficient for our purposes, but the chosen ones are simpler.

 \begin{dfn}\label{d5} A theory $T$ is called {\em absolute} if for every
closed formula $F$ an absolute (see \ref{4.2}) formula $P$ exists such that
$T\vdash \neg F\iff P$.\end{dfn}

 Replacing (4.2.3) in $A$ with {\em Church's thesis} (CT): $\exists k\forall
n:\ \beta (n)= U(k,n,\downarrow)$ (where U is a universal recursive function)
one gets the theory of recursive realizability of S.C. Kleene. It is a known
example of an absolute theory. Our theory AI is not, of course, absolute,
inasmuch as (CT) is independent of it, and can not be reduced to an absolute
formula. Then to get an absolute theory one needs an axiom implying either (CT)
or $\neg$(CT). It turns out that this is sufficient as well:

 \begin{lem}\label{l4} For any closed formula $F$, four absolute formulas $P_1,
P_2, P_3, P_4$ exist such that these statements are deducible in AI:
\begin{enumerate} \item $\neg\neg (P_1\vee P_2\vee P_3\vee P_4)$; \item
$P_1\Rightarrow \neg\neg F$; \item $P_2\Rightarrow\neg F$; \item
$P_3\Rightarrow(\neg F\iff (CT))$; \item $P_4\Rightarrow(\neg F\iff\neg(CT))$
\end{enumerate}\end{lem}

 (CT) is a very strong axiom. It excludes any non-recursive sequences, e.g.
random ones. The axiom $\neg (CT)$ is, inversely, very weak. But, unexpectedly,
AI+$\neg (CT)$ {\em is also absolute:}

 \begin{thr}\label{t3} The absolute closed theorems of $AI+\neg (CT)$ are the
same as ones of the classical first order arithmetic. No essential (i.e.,
containing new theorems of the form $\neg F$) extension of $AI+\neg (CT)$ has
this property.\end{thr}

 Thus $AI+\neg (CT)$ is a maximal conservative extension of classical
arithmetic and relatively to it is, in a sense, consistent and complete. The
basic goal of introducing this theory was to study the effects of the axiom
schema (IP).

\paragraph{Proof of Theorem~\ref{t3}:} We need for each closed formula $F$ to
establish a corresponding absolute formula $\ov F$ such that: $(AI+ \neg
(CT))\vdash \neg F\iff \ov F$, and if $F$ is absolute itself, then $\neg F\iff
\ov F$ is deducible in first order arithmetic. Besides, we need to show that
all deduction rules and axioms of $(AI+\neg (CT))$ will be converted into
derivative deduction rules and theorems of first order arithmetic. We shall
indicate the transformation $\neg F$ into $\ov F$ and explain its meaning
without writing out all routine formal deductions. Due to Proposition~\ref{p6}
one may restrict himself to formulas of the kind $F=\forall\alpha\exists\beta
P(\alpha,\beta)$, where $P$ is absolute. We say that $F$ is {\em rejected} on
$\gamma\in\Omega$ if for any recursive function $r:\N\to \N$ it is false that
for any recursive operator $k:\Omega\to\Omega$, applicable to $\gamma$,
$P(\alpha,\beta)$ holds, where $\alpha= k(\gamma)$, and $\beta= k'(\gamma)$,
$k'= r(k)$. Let $\mu$ be a recursive continuous measure. It turns out that the
equivalence of $\neg F$ to the formula ``$F$ is rejected for $\mu$-almost all
$\gamma$" is deducible in $AI+\neg (CT)$.

 The latter formula can be written in an absolute form and chosen as $\ov F$.
The point is that the quantifier ``for almost all $\gamma$" in contrast to the
quantifier ``for all $\gamma$" is expressible in the first order language.
Obviously the formula ``$F$ is rejected on $\gamma$", being absolute, can be
presented in the form of $\forall n_k\neg\forall n_{k-1}\neg...\forall n_0\neg
R(\gamma,n_0,n_1,...,n_k)$, where $R$ is a recursive predicate, monotonic on
each of the arguments $n_i$ (up -- for the even $i$ and down -- for the odd
ones). Let us show by means of recursion, how the predicate $\mu
\{\gamma:\forall n_{i-1}\neg\forall n_{i-2}\neg...\forall n_0\neg
R(\gamma,n_0...n_k)\}\ \ge r$ is expressed by an absolute formula. For i=0 it
is trivial. Now let, at the given i, the predicate be expressed in the form of
$S_i(r,n_k...n_i)$. Then $\forall n_i\forall r'>(1- r)\ \neg S_i(r',n_k...n_i)$
can serve as $S_{i+1}(r,n_k...n_{i+1})$. Thus, it remains to show that $\neg F$
is equivalent in $AI+\neg (CT)$ to the assertion ``$F$ is rejected for
$\mu$-almost all $\gamma$".

 \begin{lem}\label{l5} Let $\mu$ and $\mu$'\ be r.e. measures and $\mu$ be
continuous. Then recursive deterministic reciprocal in the domains operators
$P$ and $P'$ on $\Omega$ exist defined on $\mu$'\ (resp. $\mu$)-almost all
non-recursive sequences such that $\mu'=P(\mu)$.\end{lem}

 The proof of this lemma follows from Theorem 3.1 (b) in [Zvonkin, Levin].
Since the property ``$F$ is rejected on $\gamma$" is invariant with respect to
any recursive reversible transformation of $\gamma$, it is sufficient to prove
the equivalence of $\neg F$ to ``$F$ is rejected for $B$-almost all $\gamma$"
where $B$ is the uniform measure on $\{0,1\}^\N$. By virtue of the same
invariance and 0-1 law ([Kolmogorov 33]), the set $A$ of all $\gamma$, on which
$F$ is rejected, can be only of measure 0 or 1 with respect to $B$. Hence if
$R$ is the set of all recursive sequences, the measure of $(A\cap\neg R)$ or of
$(\neg A\cap\neg R)$ equals 0 with respect to any other recursive $\mu$ as
well. Then by virtue of Theorem~\ref{t2} a sequence exists (and it can easily
be defined by an absolute formula), upon which all complete $\gamma$ from this
set depend. The axioms of $AI+\neg (CT)$ imply that any universal sequence
(from axiom (4.2.3)) is non-recursive, equivalent to a complete one, and
independent of sequences defined by absolute formulas. Therefore in the case
$\mu (A)=0$, $F$ is not rejected on a universal $\gamma$ and $\neg\neg F$
holds. In the opposite case $\neg F$ holds by analogous reasons. These
reasonings can be easily transformed to formal proofs in $AI+\neg (CT)$. Each
of the two cases gives implication in one of the directions between $\neg F$
and ``$F$ is rejected for $\mu$-almost all $\gamma$". \qed

\section{\label{5} Theory of Turing Degrees}
	\subsection{\label{5.1} Independence and Negligible Sets}

 A natural field for application of Algorithmic Information Theory is the
Theory of Turing Degrees. One may interpret the recursive reducibility of
$\alpha$ to $\beta$ as that $\beta$ contains all (but finite amount of)
information about $\alpha$. However, the informational concepts are subtler and
more elegant than reducibility degrees. In particular, they are invariants
applicable to finite objects as well ( Proposition~\ref{p2} shows that \I(x,y)
is left invariant within a constant by any recursive reversible transformation
of \N).

 One of the new possibilities is the informational approach to the concept of
Independence in addition to the concept of Reducibility. In terms of
reducibility degrees one can also say that $\alpha$ and $\beta$ are independent
if any sequence reducible to both of them is trivial (recursive). But a simple
example shows that this definition is not adequate to intuition. Let $\alpha$
and $\gamma$ be random 0,1-sequences, of independent trials, with the
probability of $\alpha_n= 0$ be $1/2$, and the probability of $\gamma_n= 0$ be
0.99. Let $\beta_n=\alpha_n\oplus \gamma_n$. Then, $\alpha$ and $\beta$ are
almost always such that only recursive sequences are reducible simultaneously
to both of them, though 99 percent of the (random) contents of $\alpha$ and
$\beta$ coincide (so it is hard to consider them independent).

 Many exotic types of Turing degrees are known. Such are, e.g., ``minimal"
degrees containing indivisible information (any part of the information of such
degree $\beta$, i.e., a degree $\alpha <\beta$, is equivalent to 0 or to
$\beta$). The existence of such degrees is proven by diagonal methods. Such
sequences can not appear in any combination of random and recursive processes
(see [Rogers]). One may hope that many complications of the Theory of Turing
Degrees are caused by exotic examples of this kind, and the theory of
``realistic degrees" is simpler. We shall see that this is only partially so.

 Let us use the concept of Independence to define the notion of ``negligible
sets" of sequences and study properties of Turing degrees ``to within this
negligibility". A set $A\subset\Omega$ is called {\em inaccessible,} if its
complement is closed with respect to any recursive operator $F$ (i.e.,
$a\mathrel{\not\in} A\ \Rightarrow\ F(\alpha)\mathrel{\not\in}A$).

 \begin{prp}\label{p7} The following properties of a set $A\subset\Omega$ are
equivalent:\begin{enumerate}
 \item A sequence $\alpha\in\Omega$ exists on which all $\beta\in A$ are
dependent (i.e., $\exists\alpha\forall\beta\in A:\ \I(\alpha:\beta)=\infty$).
 \item A is a subset of an inaccessible set, whose any (or some continuous)
r.e. measure is 0.
 \item $\ov\M(A)=0$. \end{enumerate}\end{prp}

 \paragraph{Proof:} (1)$\iff$ (3) follows from the Note~\ref{n3}. It is obvious
that F(\M), the image of {\M} at an arbitrary recursive mapping
$F:\Omega\to\Omega$, is an r.e. semimeasure and hence $F(\M)\preceq \M$.
Therefore, if $\ov\M(A)= 0$, then $\cup F^{-1}(A)= A_{1}\supset A$ is
inaccessible and $\ov\M(A_1) =0$. This gives (3)$\Rightarrow$ (2).
Lemma~\ref{l5} implies equivalence of ``some" and ``any" in (2). Any r.e.
semimeasure is the image of a recursive measure at a recursive mapping
$\Omega\Rightarrow\Omega$ (see [Zvonkin, Levin], section~3.2). This gives
(2)$\Rightarrow$ (3). \qed

 The sets with any of these three properties are called {\em negligible} (this
neglect is, of course, based on the belief in the Independence Postulate). Two
sets $A$ and $B$ are called {\em i-equivalent} if their symmetric difference is
negligible. ``A property of Turing degrees" means a Turing-invariant set
$A\subset\Omega$. Studying them to within i-equivalence is simpler, since some
properties of ``exotic" degrees are excluded. The Boolean algebra of Borel sets
of Turing degrees is denoted by $K$, and $L$ is its factor-algebra with respect
to the i-equivalence.

\subsection{\label{5.2} Types of Turing Degrees}

 In ( \ref{2.3}) the concept of ``sequence completeness" was considered. The
set of incomplete sequences has a property very close to negligibility. Namely,
Item 2 in Proposition~\ref{p7} is obtained from Item 1 of Proposition~\ref{p4}
by omitting the word ``total". Thus, incomplete sequences cannot arise in a
process running in time bounded by a total recursive operator. Let us call {\em
regular} a sequence Turing-equivalent to a complete one. It is natural to
consider properties of the Turing degrees of regular sequences. It turns out
that only four of them are not equivalent.

 \begin{thr}\label{t4} Any Turing-invariant Borel set of regular sequences is
i-equivalent to:\begin{enumerate}
 \item The empty set or
 \item The set of recursive sequences or
 \item The set of all regular sequences or
 \item The set of all regular non-recursive sequences.\end{enumerate}\end{thr}

 Thus, the properties of a regular sequence (to within i-negligible sets)
depend only on its recursiveness, and these sequences form the two most natural
elements (atoms) of the algebra L.

 \paragraph{Proof:} As it follows from Lemma~\ref{l5}, any set $A$ of
non-recursive sequences, invariant with respect to Turing equivalence, either
is of measure 0 at any recursive measure $\mu$, or (for any $\mu$) contains
$\mu$-almost all non-recursive sequences. Then, by virtue of Theorem~\ref{t2},
a $\beta$ exists such that all the complete non-recursive sequences either from
A, or from the complement of A, respectively, depend on $\beta$. And $A$ is
i-equivalent to one of the four sets, mentioned in Theorem~\ref{t4}, since any
invariant set contains either all recursive sequences, or none. \qed

 Other Turing degree types consist of non-regular sequences. It is difficult
even to prove that their union is not negligible. Nevertheless $L$ contains
(see [V'yugin], [Levin, V'yugin]) an infinitely divisible element and a
countable number of atoms. Only two of them (namely, 2 and 4 of
Theorem~\ref{t4}) contain complete sequences.

\paragraph {ACKNOWLEDGMENTS} I am deeply grateful to A.R. Meyer for
improvements to this text, for financial and moral support, and to P. Elias and
G.E. Sacks who agreed to read the first draft of the work. My wife Larissa
actually wrote this paper (before I spoiled it). Many colleagues had a hard
time convincing me to omit my bravest discoveries in English. I am grateful to
all of them and sorry that their mission was so difficult.


\begin{thebibliography}{99}

\bibitem{1} A. Church. On the Concept of Random Sequence. {\em Bull. Amer.
Math. Soc.,} 46:254-260, 1940.

\bibitem{2} G.J. Chaitin. On the Length of Programs for Computing Finite Binary
Sequences. {\em J. Assoc. Comput. Math.,} 13:547-570, 1966. Part II, 15, 1968.

\bibitem{3} G.J. Chaitin. A Theory of Program-Size Formally Identical to
Information Theory. {\em Journal ACM,} 22:329-340, 1975.

\bibitem{4} P. Gacs. On the Symmetry of Algorithmic Information.
 {\em Soviet Math. Dokl.,} 15:1477, 1974.

\bibitem{5} D. Hilbert. Mathematical Problems.
 {\em Bull. Amer. Math. Soc.,} 2(8):437-479, 1902.

\bibitem{6} S.C. Kleene and R.E. Vesley. {\em The Foundations of Intuitionistic
Mathematics.} 1965, N.-Holland Publ. Co., Amsterdam.

\bibitem{7} S.C. Kleene. {\em Mathematical Logic.} 1967, J. Wiley \& Sons,
Inc., New York.

\bibitem{8} A.N. Kolmogorov. {\em Grundbegriffe der Wahrscheinlichkeitrechung.}
1933, Berlin. (The 2nd Russian Edition {\em Osnovnye Poniatija Teorii
Verojatnostej,} Nauka, Moscow, 1974).

\bibitem{9} A.N. Kolmogorov. Three Approaches to the Concept of The Amount of
Information. {\em Probl. Inf. Transm.,} 1(1), 1965.

\bibitem{10} A.N. Kolmogorov. Talk Resume.
 {\em Uspekhi Mat. Nauk,} 2:201, 1968.

\bibitem{11} A.K. Zvonkin and L.A. Levin. The Complexity of finite objects and
the Algorithmic Concepts of Information and Randomness.
 {\em Russian Math. Surveys,} 25(6):83-124, 1970.

\bibitem{12} L.A. Levin. On the notion of a Random Sequence.
 {\em Soviet Math. Dokl.,} 14(5):1413, 1973.

\bibitem{13} L.A. Levin. Universal Sequential Search Problems.
 {\em Probl. Inf. Transm.,} 9(3):265-266, 1973.

\bibitem{14} L.A. Levin. On Storage Capacity for Algorithms.
 {\em Soviet Math. Dokl.,} 14(5):1464-1466, 1973.

\bibitem{15} L.A. Levin. Laws of Information Conservation (non-growth) and
aspects of the Foundations of Probability Theory.
 {\em Probl. Inf. Transm.,} 10(3):206-210, 1974.

\bibitem{16} L.A. Levin. On the Principle of Conservation of Information in
Intuitionistic Mathematics. {\em Soviet Math. Dokl.,} 17:601-605, 1976.

\bibitem{17} L.A. Levin. Various Measures of Complexity for finite objects
(Axiomatic Descriptions). {\em Soviet Math. Dokl.,} 17(2):522-526, 1976.

\bibitem{18} L.A. Levin. Uniform Tests of Randomness.
 {\em Soviet Math. Dokl.,} 17(2):337-340, 1976.

\bibitem{19} L.A. Levin. and V.V. V'yugin. Invariant Properties of
Informational Bulks. {\em Lecture Notes on Computer Science,}
 53:359-364, Springer, 1977.

\bibitem{20} A.A. Markov. On Normal Algorithms which Compute Boolean Functions.
{\em Soviet Math. Dokl.,} 5:922-924, 1964.

\bibitem{21} P. Martin-Lof. The Definition of Random Sequences.
 {\em Inform. Contr.,} 9:602-619, 1966.

\bibitem{22} M.L. Minsky. Problems of Formulation for Artificial Intelligence.
{\em in Proc. Symp. in Applied Math.,} 14, 1962, Am. Math. Soc.

\bibitem{23} R. von Mises. and H. Geiringer. {\em The Mathematical Theory of
Probability and Statistics.} 1964, Academic Press, N.Y.

\bibitem{24} H. Rogers. {\em Theory of Recursive Functions and Effective
Computability.} 1967, New York.

\bibitem{25} C.P. Schnorr. {\em Zufaelligkeit und Wahrscheinlichkeit.} Lecture
Notes in Math., 218, 1971, Springer.

\bibitem{26} R.J. Solomonoff. A Formal Theory of Inductive Inference. {\em
Inform. Contr.,} 7(1):l-22, 1964.

\bibitem{27} J. Ville. {\em Etude critique de la concept de Collectiff.} 1939,
Gauthier-Villars, Paris.

\bibitem{28} V.V. V'yugin. The Algebra of Invariant Properties of Binary
Sequences. {\em Probl. Inf. Transm.,} 18(2):147-161, 1982.

\end{thebibliography}
  \end{document}